\documentclass[11pt,twoside,a4paper]{article}
\usepackage{amsmath}
\usepackage{amsfonts}
\usepackage{amsthm}
\usepackage{etoolbox}
\usepackage{graphicx}
\usepackage{changepage,threeparttable}
\usepackage[linesnumbered,ruled,vlined]{algorithm2e}
\usepackage{xcolor}
\usepackage{booktabs, multirow} 
\usepackage{graphicx}
\usepackage{url,cite}
\usepackage{mathdots}
\usepackage[cm]{fullpage}
\makeatletter
\patchcmd{\maketitle}{\@fnsymbol}{\@alph}{}{}  
\makeatother
\newtheorem{definition}{Definition}
\newtheorem{theorem}{Theorem}[section]
\newtheorem{lemma}{Lemma}[section]
\newtheorem{remark}{Remark}
\SetKwInput{KwInput}{Input}                
\SetKwInput{KwOutput}{Output}              
\title{Minors solve the elliptic curve discrete logarithm problem}
\author{Ansari Abdullah\thanks{SPPU, Pune, India} \and Ayan Mahalanobis\thanks{IISER Pune, Pune, India}}
\date{}
\begin{document}
\maketitle
\begin{abstract}
The elliptic curve discrete logarithm problem is of fundamental importance in public-key cryptography. It is in use for a long time. Moreover, it is an interesting challenge in computational mathematics. Its solution is supposed to provide interesting research directions.

In this paper, we explore ways to solve the elliptic curve discrete logarithm problem. Our results are mostly computational. However, it seems, the methods that we develop and directions that we pursue can provide a potent attack on this problem.  This work follows our earlier work, where we tried to solve this problem by finding a zero minor in a matrix over the same finite field on which the elliptic curve is defined. This paper is self-contained.
\end{abstract}
\section{Introduction} Let $\mathbb{F}$ be a finite field of arbitrary characteristic. Let
\begin{equation}\label{mat1}
  \mathcal{A}=\begin{pmatrix}
    a_{11} & a_{12} & a_{13} & \ldots & a_{1d}\\
    a_{21} & a_{22} & a_{23} & \ldots & a_{2d}\\
    \vdots & \vdots & \vdots & \ddots & \vdots\\
    a_{l1} & a_{l2} & a_{l3} & \ldots & a_{ld}
  \end{pmatrix}
\end{equation}
be a $l\times d$ matrix over $\mathbb{F}$. Let $\alpha,\beta$ be two ordered subsets of $\{1,2,\ldots,l\}$ and
$\{1,2,\ldots,d\}$ respectively and are of the same size $k$. Then we define the sub-matrix $\mathcal{A}[\alpha|\beta]$ of $\mathcal{A}$ to be the square matrix of size $k$ consisting of elements that are in the intersection of the rows in $\alpha$ and columns in $\beta$ in $\mathcal{A}$. The ordering in the sub-matrix is the same as the ordering in the matrix $\mathcal{A}$. The determinant of $\mathcal{A}[\alpha|\beta]$ is a minor of $\mathcal{A}$.

A non-singular elliptic curve over $\mathbb{F}$ will be denoted by $\mathcal{E}$. We will consider it as a plane projective curve in the projective plane $\mathbb{P}^2$.  It is well known that there is an abelian group operation on $\mathcal{E}$ with the point at infinity $\mathcal{O}$ as the zero element. We will use that fact and denote the abelian group by $\mathcal{E}$ as well. We will often assume that $\mathcal{E}$ is of prime order and the order is $p$, a prime. The reader is reminded, if $\mathcal{E}$ is not of prime order but $p$ is any prime divisor of the order of $\mathcal{E}$, a subgroup of $\mathcal{E}$ exist which is of order $p$. We will take that subgroup to be the $\mathcal{E}$ in this paper. Once we assume that the group of rational points of an elliptic curve is of order $p$, the next thing to assume is that it has a generator $\mathrm{P}$, i.e., $\mathcal{E}=\langle \mathrm{P}\rangle$. 
\begin{definition}[The elliptic curve discrete logarithm problem]
Let $\mathcal{E}$ and $\mathrm{P}$ be as above. Then the elliptic curve discrete logarithm problem is: given $\mathrm{P}$ and $\mathrm{Q}=m\mathrm{P}$ in $\mathcal{E}$, where $1<m<p$, compute $m$.
\end{definition}
In this paper, we will deal with the elliptic curve discrete logarithm problem exclusively. Since there is no scope of confusion, we will call it the \emph{discrete logarithm problem}.

We write this paper to share a few things:
\begin{description}
\item[a)] Finding a zero minor in an appropriate matrix, just like $\mathcal{A}$, solves the discrete logarithm problem.
\item[b)] We share some success in finding a zero minor with experiments involving almost principal minor (see Table~\ref{tbl:all-apm}). We believe the reason behind this success is the existence of a \emph{set of initial minors}. This solidifies our conjecture of initial minors.
\item[c)] We solve the discrete logarithm problem for many different large cases. Moreover, we claim that there is a successful attack in the making, using methods in this paper.
\end{description} 
\section{Generate $\mathcal{A}$}
The central idea behind constructing $\mathcal{A}$ is to look at a matrix $\mathcal{M}$ and construct its left-kernel $\mathcal{K}$.
Recall that the left-kernel is the kernel of the transpose of a matrix. Then $\mathcal{A}$ is just a sub-matrix of $\mathcal{K}$.
The idea behind constructing $\mathcal{M}$ is a well known theorem. Before we go to the theorem, we need to be familiar with the concept of a divisor on $\mathcal{E}$.
\begin{definition}[Divisor]
Let $\mathcal{E}$ be the elliptic curve defined over $\mathbb{F}$. The divisor $D$ is defined as the formal sum $D=\sum_{i}n_i[P_i]$, where $n_i$ are integers and $P_i$ are rational points on $\mathcal{E}$. Since the field $\mathbb{F}$ is finite, the formal sum is always finite.
\end{definition}
We present bare minimum details on divisor in this paper, enough to make it self contained. For more on divisors the reader can consult any text book in algebraic geometry. Fulton~\cite{fulton}, Hartshorne~\cite{hartshorne}, Milne~\cite{milne} and  Silverman~\cite{silverman} all contain such details.

If $n_i\geq 0$ the divisor is called effective. The group of all divisors $D$ is a free abelian group and is called the divisor class group. When the divisors belong to the elliptic curve $\mathcal{E}$ it is denoted by $\text{Div}(\mathcal{E})$. The degree of the divisor $D$ is $\sum_{i} n_i$. The set of all degree zero divisors form a subgroup of the group of divisors and is denoted by $\text{Div}^0(\mathcal{E})$.

Let $f$ be a homogeneous rational function in $\mathbb{P}^2$. Then $f=g/h$ where $g$ and $h$ are homogeneous polynomials in three variables of the same degree and furthermore $\mathcal{E}$ does not divide $h$. Then $f$ defines a rational function on $\mathcal{E}$. The points of intersection of $g$ are zeros of $f$ (counting multiplicity) and the points of intersection of $h$ and $\mathcal{E}$ are poles of $f$ (counting multiplicity). Over a algebraically closed field, the number of zeros and poles of $f$ are finite and are the same. Then one can define the divisor corresponding to $f$ as $D_f = \sum n_i[P_i]$ where $g$ intersects $\mathcal{E}$ at $P_i$ with intersection multiplicity $n_i$ or $h$ intersects $\mathcal{E}$ at $P_i$ with multiplicity $n_i$. Note, the multiplicity of a zero is considered positive and that of a pole negative. These divisors are called principal divisors and denoted by $\text{Prin}(\mathcal{E})$. We define the Picard group $\text{Pic}^0(\mathcal{E})=\text{Div}^0(\mathcal{E})/\text{Prin}(\mathcal{E})$. One can now define a map from the group of rational points of an elliptic curve $\mathcal{E}$ to the group $\text{Pic}^0(\mathcal{E})$. This map is know as the Abel-Jacobi map.

\begin{equation}\mathcal{P}\mapsto [\mathcal{P}]-[\mathcal{O}]\end{equation}
For a non-singular elliptic curve over an algebraically closed field this map is an isomorphism.

Let $\mathcal{C}$ be a projective plane curve of degree $n^\prime$ not containing $\mathcal{E}$. Then according to Bezout's theorem the curve $\mathcal{C}$ can intersect $\mathcal{E}$ in at most $3n^\prime$ points (counting multiplicity). We use $d=3n^\prime$ for the rest of the paper.

We now state and sketch a proof of a theorem that was proved in Mahalanobis et al.~\cite{first} (see also~\cite{second}). We include a short proof for the convenience of the reader.

\begin{theorem}
Let $\mathcal{E}$ be as defined above. Then there is a plane projective curve $\mathcal{C}$ passing through $d$ points, $P_1,P_2,\ldots,P_{d}$ of $\mathcal{E}$ and not containing $\mathcal{E}$ if and only if $\sum_{i=1}^{d} P_i =\mathcal{O}$.  Where $\mathcal{O}$ is the neutral element of the abelian group $\mathcal{E}$.
\end{theorem}
\begin{proof}
The proof over algebraically closed field is straightforward consequence of the fact that the Abel-Jacobi map is an isomorphism.
We pay special attention to the part where there is a plane projective curve $\mathcal{C}$ of degree $n^\prime$ that intersects the elliptic curve $\mathcal{E}$ at $d$ points. We construct the rational function $C/z^{n^\prime}$. Now this rational function has a pole of order $d$ at $\mathcal{O}$. Thus the Abel-Jacobi map says that $\sum_iP_i=\mathcal{O}$. The reverse direction is straightforward.

Extending the proof to finite fields uses a well known trick using the Frobenius automorphism of a finite field.
\end{proof}
Coming back to the discrete logarithm problem, we need to compute $m$ from $\mathrm{P}$ and $\mathrm{Q}$. We intend to compute unique points of the form $n_i\mathrm{P} = P_i$ and $-n_j\mathrm{Q} = Q_j$ where $n_i$ and $n_j$ are integers in the range $1$ to $p$, such that, $\sum_{i} P_i + \sum_{j} Q_j=\mathcal{O}$. This gives us $\sum_{i}n_{i} - m\sum_j n_j = 0 \mod p$. We know $n_i$ and $n_j$ and then we have a linear equation in one unknown $m$. Thus we can solve for $m$ and the discrete logarithm problem is solved. 

Now let us construct the matrix $\mathcal{M}$. For this we take all possible monomials in $x,y,z$ of degree $n^\prime$ ordered in a particular way. The order once chosen remain fixed throughout the paper.  It is known that the number of such monomials is ${n^\prime +2 \choose 2}$.
We construct $\mathcal{M}$ one row at a time. For each $P_i$ we evaluate each monomial at $P_i$ in the above order. The vector thus formed is the $i\textsuperscript{th}$ row of $\mathcal{M}$. We do this for all $P_i$ and $Q_j$ and $\mathcal{M}$ is created. Recall that all $P_i$ and $Q_j$ are distinct. Thus rows of $\mathcal{M}$ are indexed by $P_i$ and $Q_j$. This is the same as $n_i$ and $n_j$ indexing rows in $\mathcal{M}$. This is because $n_i$ and $-n_j$ creates $P_i$ and $Q_j$. The left-kernel of $\mathcal{M}$ is denoted by $\mathcal{K}$. By points in $\mathcal{M}$ we mean those points $P_i$ and $Q_j$ that creates $\mathcal{M}$. Recall that $\mathcal{M}$ has $d$ rows. Thus $\mathcal{K}$ has $d$ columns indexed by $P_i$ and $Q_j$.

The next lemma starts the process of an actual algorithm. To prove this lemma we use the fundamental theorem on four subspaces of a matrix $A$. This states that the nullspace of $A$ is an orthogonal complement of the column space of $A^T$ and the nullspace of $A^T$ is an orthogonal complement of the column space of $A$.
\begin{lemma}
The left-kernel $\mathcal{K}$ is zero if and only if there is no homogeneous plane curve $\mathcal{C}$ of degree $n^\prime$ that passes through the points in $\mathcal{M}$ and does not contain the elliptic curve $\mathcal{E}$.
\end{lemma}
\begin{proof} The proof uses an easy and elementary counting argument. First assume
that $\mathcal{K}=0$. Then the dimension of the column space of $\mathcal{M}^T$ is $3n^\prime$. Then the dimension of the kernel of $\mathcal{M}$ is ${n^\prime +2 \choose 2}-3n^\prime={n^\prime -1 \choose 2}$. This proves that the right-kernel of $\mathcal{M}$ contains all monomials of degree $n^\prime -3$. This is the set of all monomials of degree $n^\prime$ that contains the elliptic curve.

Conversely assume that that the right-kernel of $\mathcal{M}$ contain curves $\mathcal{C}$ that contain the elliptic curve. Then the dimension of the subspace is ${n^\prime -1 \choose 2}$. Furthermore, ${n^\prime+2\choose 2}-{n^\prime-1\choose 2}=3n^\prime$. Thus the dimension of the kernel of $\mathcal{M}^T=0$.
\end{proof}
The above lemma is interesting in its own right. But, from the perspective of solving the discrete logarithm problem it is rather useless. This is because, to use the above theorem in solving the discrete logarithm problem we have to find points $P_i$ and $Q_j$ that adds up to $\mathcal{O}$. There is no way to predict which points might do that. So in practice the complexity is worse than solving the discrete logarithm problem by exhaustive search. 

The next theorem is of vital interest for the algorithm that we develop in this paper. In this case we take $l$ extra points. In other words, the matrix $\mathcal{M}$ created above now has $d + l$ rows constructed out of $d + l$ points $P_i$ and $Q_j$. From this matrix, the left-kernel $\mathcal{K}$ can be computed easily. The columns of this left-kernel $\mathcal{K}$ corresponds to the points $P_i$ and $Q_j$ on $\mathcal{E}$ from which $\mathcal{M}$ was created. For the experiments we take $n^\prime=\log_{2}{p}$. In practice one can take $n^\prime$ to be some constant multiple of $\log_2{p}$ where the constant is a small integer. When this is the case, computing $\mathcal{K}$ from $\mathcal{M}$ is polynomial in both space and time complexity.
\begin{theorem}{(Extra points theorem)}\label{extra}
Let $\mathcal{M}$ has $d + l$ rows coming from $d + l$ points $P_i$ and $Q_j$ for some integer $l\geq 1$. Then the left-kernel $\mathcal{K}$ of $\mathcal{M}$ has dimension $l$. Furthermore, if there is a vector $v$ in $\mathcal{K}$ with $l$ zeros, there is a curve $\mathcal{C}$ that passes through $d$ points $P_i$ and $Q_j$ which corresponds to non-zero points of the vector $v$. 
\end{theorem}
\begin{proof}
First note that the new $\mathcal{M}$ has $d + l$ rows and ${n^\prime + 2 \choose 2}$ columns. It follows from Bezout's theorem, if there is a curve $\mathcal{C}$ that passes through more than $3n^\prime$ points of the elliptic curve $\mathcal{E}$, it must contain the elliptic curve as a component. This says that the dimension of the right-kernel of $\mathcal{M}$ is ${n^\prime -1 \choose 2}$. Since the right-kernel is the orthogonal complement of the column space of $\mathcal{M}^T$. This says that the dimension of the column space of $\mathcal{M}^T$ is  ${n^\prime+2 \choose 2}-{n^\prime -1 \choose 2}=3n^\prime$. Since the row-rank equals the column-rank, the dimension of the column space of $\mathcal{M}$ is $3n^\prime$. Then the orthogonal complement of the column space is the dimension of $\mathcal{K}$ the kernel of $\mathcal{M}^T$ which is $l$.

Now if there is a vector $v$ with $l$ zeros, then take the nonzero components of $v$ and the corresponding rows of $\mathcal{M}$. Then using the earlier theorem we have a curve passing through those points where $v$ is non-zero. 
\end{proof}
\begin{remark}\label{rem}
In the above theorem we are talking about $l$ zeros. Why the number of zeros can not be different than $l$?
\end{remark}
In our experience with computing the discrete logarithm problem, we have never seen this happen except in one case. This is when a chosen point $n_i$ is the secret $m$. This is tantamount to solving the discrete logarithm problem by guessing and never happens for large enough prime $p$. The other reason is theoretical. If there is a vector $v$ with more than $l$ zeros then that vector is non-zero in less than $3n^\prime$ places. We create a sub-matrix $\mathcal{M}^\prime$ of $\mathcal{M}$ corresponding to the non-zero position in $v$. The number of these positions is less than $3n^\prime$ and thus the number of rows in $\mathcal{M}^\prime$ is less than $3n^\prime$. We extend this matrix $\mathcal{M}^\prime$ to $3n^\prime$ rows by putting in arbitrary rows corresponding to arbitrary points from $\mathcal{E}$. Then the non-zero components of $v$ along with few zeros at appropriate places will be in the left-kernel of this extended matrix. Then the sum of the points that make the rows of this extended matrix will be $\mathcal{O}$. This is absurd as the rows that extend $\mathcal{M}^\prime$ was arbitrarily chosen. Thus the number of zeros never exceed $l$. However, if there are fewer than $l$ zeros then the curve $\mathcal{C}$ so obtained contain the elliptic curve $\mathcal{E}$ and is thus useless for our purpose. From the matrix $\mathcal{K}$, the matrix $\mathcal{A}$ is extracted.
\begin{remark} How do we ensure that there is atleast one $P_i$ and one $Q_j$ in the
non-zero entries of $v$ created above? 
\end{remark}
The above is an important point because unless both $P_i$ and $Q_j$ are present in the non-zero entries of the $v$ produced above, the discrete logarithm problem is not solved. However, there are many ways to solve this problem. First, in our experiments we optimize $d$ and $l$. It is well known that ${d+l}\choose{l}$ is the biggest when $d=l$. Thus we choose $d=l$ for all our experiments and for the rest of the paper. Thus our matrix $\mathcal{K}$ is of size $l\times 2l$ and $\mathcal{A}$ is of size $l \times l$. 

So now we have $2l$ points $P_i$ and $Q_j$. We have to ensure that every choice of $l$ points from $2l$ points must have at least one $P_i$ and $Q_j$. We took $l-1$ points $P_i$ and $l+1$ points $Q_j$. There is a small chance that all the $l$ points that we get from non-zero components of $v$ are all $Q_j$. In which case we have not solved the discrete logarithm problem. While doing these experiments, in many cases we went all the way and solved for $m$. And we were successful all the time. This is because ${l+1\choose l}$ is very small compared to ${2l\choose l}$. 
\section{A zero-minor solves the discrete logarithm problem}
In this section, we show that finding a zero-minor in $\mathcal{A}$ solves the discrete logarithm problem by showing that there is one-one correspondence between zero-minors in $\mathcal{A}$ and maximal zero-minors in $\mathcal{K}$.  Note that $\mathcal{K}$ is a $l\times 2l$ matrix and $\mathcal{A}$ is a $l\times l$ matrix. Recall that for the purpose of experiments we are assuming that $d=l$.

The size of $\mathcal{K}$ and $\mathcal{A}$ follows from the fact that the left-kernel $\mathcal{K}$ is a $l$ dimensional vector space. The basis matrix of the subspace $\mathcal{K}$ is denoted by $\mathcal{K}$ as well; where rows and columns of $\mathcal{K}$ are indexed by $\{1,2,\ldots,l\}$ and $\{1,2,\ldots,2l\}$ respectively. 

Each row of the matrix $\mathcal{K}$ is a basis vector of the subspace $\mathcal{K}$. Then one can do row operations to reduce $\mathcal{K}$ into the following form:
  \begin{equation}\label{mat2}
  \mathcal{K}=\begin{pmatrix}
    a_{11} & a_{12} & a_{13} & \ldots & a_{1l}& 0 & \ldots & 0 & 1\\
    a_{21} & a_{22} & a_{23} & \ldots & a_{2l}& 0 &\ldots & 1 & 0\\
    \vdots & \vdots & \vdots & \ddots & \vdots& 0 &\iddots & \vdots &\vdots\\
    a_{l1} & a_{l2} & a_{l3} & \ldots & a_{ll}& 1 &\ldots & 0 & 0
  \end{pmatrix}
\end{equation}
where the right hand side is in the reverse identity form, which we will refer to as the sparse part. The left hand side is the dense part which is $\mathcal{A}$. If any one of $a_{ij}=0$ where $i,j=1,2,\ldots,l$ then the discrete logarithm problem is solved. So we can safely assume that $a_{ij}\neq 0$ for all $i,j$.

By a maximal minor in $\mathcal{K}$ we mean a minor of size $l$, in which case all the rows are chosen and $l$ columns from $2l$ columns are chosen. Let $\{a_1,a_2,\ldots,a_k\}$ be the rows and $\{b_1,b_2,\ldots,b_k\}$ be the columns of a zero-minor in $\mathcal{A}$. Let $\{a_1^\prime,a_2^\prime,\ldots,a_{l-k}^\prime\}$ be the ordered complement of the set $\{a_1,a_2,\ldots,a_k\}$ in $\{1,2,\ldots,l\}$. Then the maximal minor in $\mathcal{K}$ consist of the columns $\{b_1,b_2,\ldots,b_k,l+a_1^\prime,l+a_2^\prime,\ldots,l+a_{l-k}^\prime\}$. To see why this is true, notice that if there is a zero-minor in $\mathcal{A}$, then there are row operations in $\mathcal{A}$ that gives a zero row in the sub-matrix of the minor. Use the same row operations on $\mathcal{K}$. Now note in the sparse part of $\mathcal{K}$ the position of the one is $(i,l+i)$ for each $i\in\{1,2,\ldots,l\}$. The row operation changes the structure of the sparse part of $\mathcal{K}$. The zero-row in the sub-matrix gives rise to a non-zero part of size $k$ in the sparse part. So the maximal minor is the columns in $\mathcal{A}$ and then the part in the sparse part of $\mathcal{K}$ that does not contain the non-zero part we just talked about, ordered accordingly. An observant reader might argue that what happens if all the rows of the sub-matrix of the minor were not used to get the zero-row. In that case we will have more than $l$ zeros in a row and then a look at Remark~\ref{rem} says that it can not be the case. So all rows must be involved. Thus we have proved the following theorem.
\begin{theorem}
For every zero-minor arising from a proper sub-matrix of $\mathcal{A}$ there is a maximal zero-minor in $\mathcal{K}$. Furthermore, if there is a maximal zero-minor in $\mathcal{K}$ then that solves the discrete logarithm problem.
\end{theorem}
\begin{proof}
The first part is already proved. To prove the last part of the theorem, first note that the maximal minor of $\mathcal{K}$ is of size $l$. Then the maximal zero-minor of $\mathcal{K}$ gives rise to a vector in the left-kernel with exactly $l$ zeros; Theorem~\ref{extra} takes care of the rest.
\end{proof}
\section{Finding a zero-minor in $\mathcal{A}$} In this section, we talk about methods and algorithms to find a zero minor in the $l\times l$ matrix $\mathcal{A}$ which is non-singular. It is easy to see that if $\mathcal{A}$ is singular the discrete logarithm problem is solved. The number of minors in $\mathcal{A}$ is exponential in the size of the matrix. For a square matrix of size $l$, just the number of principal minors is $2^{l}-1$. The number of total minors far exceed that. Thus, it is not possible to test all minors.

The purpose of our experiments in this paper is to solve this obvious conundrum -- how to talk about all minors without involving all of them? The \textit{modus operandi} for our experiments is the following:
\begin{description}
\item[$\dag$] Fix a set of minors. Then look for a zero minor in this set.
\item[$\dag$] Fix an elliptic curve over a finite field of characteristic 2 and two points $\mathrm{P}$ and $\mathrm{Q}$. 
\item[$\dag$] Create a left-kernel $\mathcal{K}$. This part is randomized. Then create the submatrix $\mathcal{A}$ from $\mathcal{K}$.  
\item[$\dag$] Check for zero minor in the fixed set of minors in $\mathcal{A}$. If one found, the program stops.
\item[$\dag$] Repeat with a new kernel $\mathcal{K}$ until a zero minor is found in the fixed set of minors and count the number of kernels needed.
\end{description}
We used three fixed sets. One, the set of all 2-minors. Two, the set that comes out of Gaussian elimination and Schur complement and the third is the set of almost principal minors. The size of all these sets are polynomial in the size of the matrix $\mathcal{A}$. We report our results and subsequent improvements.

We started by testing all $2 \times 2$ minors of $\mathcal{A}$ for small enough fields. We call this experiment the \emph{all\_2\_minor} experiment and was reported in our earlier work~\cite{second}. We provide the result for this experiment in Table~\ref{tab:2minorResult}.

\begin{table}[h]\centering
\caption{all\_2\_minor : Average number of kernels to solve ECDLP}
\label{tab:2minorResult}
\begin{tabular}{lccc}\toprule
$\mathbb{F}_q$ & Average kernel count & $log_2 (p)$  & No. of rows in $\mathcal{A}$ \\\midrule
$2^{40}$ & 113.4 & 39 & 360 \\
$2^{43}$ & 716.3 & 42 & 387 \\
$2^{46}$ & 4406.26 & 45 & 414 \\
\bottomrule
\end{tabular}
\end{table}
In this experiment we counted the number of kernels required before a zero minor was detected by checking all 2 minors. For each input the experiment was repeated forty times and the average kernel count was reported in Table~\ref{tab:2minorResult}.
Here the average kernel count increases exponentially.
The next obvious step was to test all $3 \times 3$ minors.
But there are too many of them and it was computationally very expensive to test all 3 minors. We can not keep the dimensions of the sub matrix increasing. That will defeat the purpose of the experiment.

\subsection{Initial Minors}The purpose of this paper, is to show experimentally that probably there is a small set of minors and it is enough to check that set for a zero minor. We call this the initial minor conjecture and is from our earlier work~\cite{second}. This is the reason that we are experimenting with a fixed set of minors. The idea is to somehow approximate this set. Gain some understanding about it before we can formally prove its existence.

\textbf{Conjecture (Initial Minors).} \textit{For a matrix $\mathcal{A}$,  there exists a set of minors, whose cardinality is bounded by a polynomial or a sub-exponential function in the size of the matrix. If all the minors in this set is non-zero then all minors of the matrix is non-zero as well. This set is called a set of initial minors.}

We report some success with our experiments. We attribute this success to the existence of a set of initial minors. Our first experiment was reported~\cite{second} and was the Gaussian elimination Schur complement experiment. We present that result in
Table~\ref{tab:gesc}.
\begin{table}[h]\centering
\caption{Gaussian elimination Schur complement: Average kernels}\label{tab:gesc}
\normalsize
\begin{tabular}{lccc}\toprule
$\mathbb{F}_q$ & Average kernel count & $log_2 (p)$  & No. of rows in $\mathcal{A}$ \\\midrule
$2^{40}$ & 2.8 & 39 & 360 \\
$2^{41}$ & 3.1 & 40 & 369 \\
$2^{42}$ & 6.6 & 41 & 378 \\
$2^{43}$ & 11.3 & 42 & 387 \\
$2^{44}$ & 17.1 & 42 & 396 \\
$2^{45}$ & 53.8 & 44 & 405 \\
$2^{46}$ & 40.2 & 45 & 414 \\
$2^{47}$ & 126 & 46 & 423 \\
$2^{48}$ & 186.6 & 47 & 432 \\
$2^{49}$ & 367.8 & 48 & 441 \\
$2^{50}$ & 887.5 & 49 & 450 \\
\bottomrule
\end{tabular}
\end{table}
The discrete logarithm problem for a $50$-bit binary field was solved using this algorithm. This algorithm works by reducing the input matrix $\mathcal{A}$ to $\mathcal{A'}$ using row reduction. The input matrix $\mathcal{A}$ is partitioned into blocks $E, F, G, H$. If $\det(E)=0$ then the problem is solved. So we assume that $E$ is non-singular. The block matrix $E^\prime$ is produced by doing a Gaussian elimination to reduce it to a upper triangular matrix. All the column entries below that upper triangular matrix is made zero using row-operations. Thus the new matrix formed is $\mathcal{A}^\prime$ is of the form:
\begin{equation}
\mathcal{A} = \begin{pmatrix}
  E & F\\
  G & H
\end{pmatrix}
\hspace{2ex} \longrightarrow  \hspace{2ex}
\mathcal{A}^\prime = \begin{pmatrix}
  E^\prime & F^\prime\\
  0 & H^\prime
\end{pmatrix}
\end{equation}
In this experiment we changed the block size of $E$. The matrix $H^\prime$ is called the Schur complement. 
The algorithm then searches for a $2 \times 2$ zero minor in $H^\prime$. The procedure is repeated with a new kernel until a zero-minor is found. If a $2 \times 2$ zero minor is found then the discrete logarithm problem is solved.

When a zero minor is found in the Schur complement, that results in a zero minor in the original matrix $\mathcal{A}$ by adjoining the indices of the zero minor in $H^\prime$ with the indices of $E$. Note that $E$ is a principal minor with contiguous rows starting with the first index of $\mathcal{A}$. More about this algorithm is in our earlier work~\cite[Section 4]{second}. 

If we compare Table~\ref{tab:2minorResult} and Table~\ref{tab:gesc}, it is clear that Gaussian elimination Schur complement made significant improvements over the the all\_2\_minor experiment. Though we are still dealing with exponential growth with the number of kernels used to get the first zero minor, this method was able to solve the discrete logarithm problem for a $50$-bit binary input and it checks a very small subset of all minors.

This success indicates that probably Gaussian elimination Schur complement was able to find a subset of a initial minor set. However the growth in the number of kernel used shows that the subset was probably not a significant subset. Motivated by the Gaussian elimination Schur complement algorithm, in the next section we did the \emph{almost principal minor experiment}.

\section{Almost principal minor algorithm}
In the Gaussian elimination Schur complement experiment, a $2 \times 2$ zero minor from $H^\prime$ was augmented with the principal block matrix $E^\prime$.
This augmented minor was then used to solve the discrete logarithm problem.
If the augmented zero minor had $k$ rows and columns, the block matrix $E^\prime$ has dimensions $(k-2) \times (k-2)$. The block matrix $E^\prime$ has contiguous row and column indices. Moreover these row and column indices are identical. Such a minor is referred to as a principal minor.

Thus a $k \times k$ minor of $\mathcal{A}$ tested in Gaussian elimination Schur complement has row indices $\alpha$ and column indices $\beta$  which are as follows:
\begin{equation}
    \alpha = \{ 1,2,3,...,(k-2), a_1, a_2\}
    \hspace{2ex}
    \beta = \{1,2,3,...,(k-2), b_1, b_2\}
\end{equation}
where, $a_1, a_2$ and $b_1, b_2$ are row and column indices respectively of a $2 \times 2$ zero minor from $H^\prime$. We say that this $k \times k$ zero minor has \emph{almost contiguous row and column indices}. The row and column indices $\alpha$ and $\beta$ respectively are from the $2 \times 2$ zero minor of $H^\prime$ are called the \textbf{deviations to the principal minor} $\{1,2,\ldots,(k-2)\}$. We refer to this $k \times k$ minor with contiguous indices followed by deviations as an \textbf{almost principal minor}.

Note that these deviations can occur either before or after the contiguous part. Some examples of almost principal minors are:
\begin{equation}
    \alpha = \{ 1, 2, 3, 4, 5, \textcolor{red}{13, 32} \}
\label{rowDev}
\end{equation}
\begin{equation}
    \beta = \{ \textcolor{blue}{1, 2}, 7, 8, 9, 10, 11 \}
\label{colDev}
\end{equation}
Also the deviations can be both before or after the contiguous part as the following example shows:
\begin{equation}
    \alpha = \{ \textcolor{blue}{1}, 7, 8, 9, 10, 11, \textcolor{red}{13} \}
\label{mixDev}
\end{equation}
These deviations can not occur within the contiguous part.

In the Gaussian elimination Schur complement experiment, the principal minor at the top left of the matrix $\mathcal{A}$ was used.
There were row operations involved in the process of converting $\mathcal{A}$ to $\mathcal{A}^\prime$. However, for an almost principal minor algorithm there are no row-operations involved. We test all almost principal minors with two and three deviation. An almost principal minor has two parameters, the principal minor and the deviation. Experimentally we observed that increasing the dimension of the principal minor and keeping deviation to two was not helpful. But, fixing the principal minor dimension to two and using two deviations followed by three deviations was helpful. 

\subsection{Experiments with almost principal minor}
This experiment used the fixed set of almost principal minor to search for a zero minor.
The input to this algorithm consists of a matrix $\mathcal{A}$ as described earlier.
These matrices are over binary fields of sizes from 25 to 47 bits.
For each input we count the number of kernels required before the first zero minor is found. Then the experiment is repeated ten times. The size of $\mathcal{E}$ is a prime $p$ and about the same size of the field.
The average number of kernels required to solve the discrete logarithm problem is then recorded for each binary field. Averaging over ten attempts is used to minimize the effect of randomization, as a randomized algorithm can sometimes solve the problem quickly or take a longer time for the same input. If the number of kernels grows gradually as inputs size increase, then ten attempts should be sufficient to minimize randomization effects. This was the case in our experiments, thus averaging  over ten attempts is effective.

Table~\ref{tbl:apm-avg-kernel-cnt} gives the average kernel count for this experiment over ten tries.
For inputs between 28 and 44 bits, the problem was solved using one kernel on average. For the 45-bit input, the average iteration count was 1.2, where one instance required three kernels. The probabilistic nature of the Las Vegas algorithm means that an occasional instance may require more kernels than the average. However we did not observe any sudden or abrupt spike. This says that the computation is stable.
\begin{adjustwidth}{-2.5 cm}{-2.5 cm}\centering\begin{threeparttable}[!htp]
\caption{APM : average kernel count}\label{tbl:apm-avg-kernel-cnt}
\footnotesize
\begin{tabular}{c|ccccccccccccccccccccc}\toprule
\multirow{2}{*}{Attempts}&  \multicolumn{20}{c}{Field size in bits} \\
 &28 &29 &30 &31 &32 &33 &34 &35 &36 &37 &38 &39 &40 &41 &42 &43 & 44 &45&46&47\\\midrule
1  &1 &1 &1 &1 &1 &1 &1 &1 &1 &1 &1 &1 &1 &1&1&1 &1&1&1&1 \\
2  &1 &1 &1 &1 &1 &1 &1 &1 &1 &1 &1 &1 &1 &1&1&1 &1&1&1&1 \\
3  &1 &1 &1 &1 &1 &1 &1 &1 &1 &1 &1 &1 &1 &1&1&1 &1&1&\textbf{3}&\textbf{3} \\
4  &1 &1 &1 &1 &1 &1 &1 &1 &1 &1 &1 &1 &1 &1&1&1 &1&1&\textbf{2}&1 \\
5  &1 &1 &1 &1 &1 &1 &1 &1 &1 &1 &1 &1 &1 &1&1&1 &1&1&1&1 \\
6  &1 &1 &1 &1 &1 &1 &1 &1 &1 &1 &1 &1 &1 &1&1&1 &1&1&1&\textbf{6} \\
7  &1 &1 &1 &1 &1 &1 &1 &1 &1 &1 &1 &1 &1 &1&1&1 &1&1&\textbf{2}&1 \\
8  &1 &1 &1 &1 &1 &1 &1 &1 &1 &1 &1 &1 &1 &1&1&1 &1&\textbf{3}&1&1 \\
9  &1 &1 &1 &1 &1 &1 &1 &1 &1 &1 &1 &1 &1 &1&1&1 &1&1&1&1 \\
10  &1 &1 &1 &1 &1 &1 &1 &1 &1 &1 &1 &1 &1 &1&1&1 &1&1&\textbf{3}&\textbf{2} \\\midrule
AVG  &1 &1 &1 &1 &1 &1 &1 &1 &1 &1 &1 &1 &1 &1&1&1 &1&1.2&1.6&1.8 \\
\bottomrule
\end{tabular}
\end{threeparttable}\end{adjustwidth}
\vspace{3ex}

Table~\ref{tab:by-apm-3d-dev} shows the deviations when the discrete logarithm was solved. It was observed that as the field size increases, there are no zero minors with two deviations, and the problem is almost always solved with three deviations.
\begin{adjustwidth}{-2.5 cm}{-2.5 cm}\centering\begin{threeparttable}[!htp]
\caption{APM : deviation when ECDLP solved}\label{tab:by-apm-3d-dev}
\footnotesize
\begin{tabular}{c|ccccccccccccccccccccc}\toprule
\multirow{2}{*}{Attempts}&  \multicolumn{20}{c}{Field size in bits} \\
 &28 &29 &30 &31 &32 &33 &34 &35 &36 &37 &38 &39 &40 &41 &42 &43 &44 &45 &46 &47  \\\midrule
1  &2 &2 &2 &\textcolor{red}{3} &2 &\textcolor{red}{3} &\textcolor{red}{3} &\textcolor{red}{3} &\textcolor{red}{3} &\textcolor{red}{3} &2 &\textcolor{red}{3} &\textcolor{red}{3} &\textcolor{red}{3} &\textcolor{red}{3} &\textcolor{red}{3}&\textcolor{red}{3}&\textcolor{red}{3}&\textcolor{red}{3}&\textcolor{red}{3} \\
2  &2 &2 &2 &2 &2 &\textcolor{red}{3} &\textcolor{red}{3} &\textcolor{red}{3} &\textcolor{red}{3} &\textcolor{red}{3} &\textcolor{red}{3} &\textcolor{red}{3} &\textcolor{red}{3} &\textcolor{red}{3} &\textcolor{red}{3} &\textcolor{red}{3}&\textcolor{red}{3}&\textcolor{red}{3}&\textcolor{red}{3}&\textcolor{red}{3}\\
3  &2 &2 &2 &2 &2 &2 &\textcolor{red}{3} &\textcolor{red}{3} &2 &\textcolor{red}{3} &\textcolor{red}{3} &\textcolor{red}{3} &\textcolor{red}{3} &\textcolor{red}{3} &\textcolor{red}{3} &\textcolor{red}{3}&\textcolor{red}{3}&\textcolor{red}{3}&\textcolor{red}{3}&\textcolor{red}{3}\\
4  &2 &2 &2 &\textcolor{red}{3} &\textcolor{red}{3} &\textcolor{red}{3} &\textcolor{red}{3} &\textcolor{red}{3} &\textcolor{red}{3} &\textcolor{red}{3} &\textcolor{red}{3} &\textcolor{red}{3} &\textcolor{red}{3} &\textcolor{red}{3} &\textcolor{red}{3} &\textcolor{red}{3}&\textcolor{red}{3}&\textcolor{red}{3}&\textcolor{red}{3}&\textcolor{red}{3}\\
5  &2 &2 &2 &\textcolor{red}{3} &2 &\textcolor{red}{3} &\textcolor{red}{3} &2 &\textcolor{red}{3} &\textcolor{red}{3} &\textcolor{red}{3}&\textcolor{red}{3} &\textcolor{red}{3} &\textcolor{red}{3} &\textcolor{red}{3} &\textcolor{red}{3}&\textcolor{red}{3}&\textcolor{red}{3}&\textcolor{red}{3}&\textcolor{red}{3}\\
6  &2 &2 &2 &\textcolor{red}{3} &2 &\textcolor{red}{3} &\textcolor{red}{3} &\textcolor{red}{3} &\textcolor{red}{3} &\textcolor{red}{3} &\textcolor{red}{3} &\textcolor{red}{3} &\textcolor{red}{3} &\textcolor{red}{3} &\textcolor{red}{3} &\textcolor{red}{3}&\textcolor{red}{3}&\textcolor{red}{3}&\textcolor{red}{3}&\textcolor{red}{3}\\
7  &2 &2 &\textcolor{red}{3} &2 &\textcolor{red}{3} &\textcolor{red}{3} &\textcolor{red}{3} &\textcolor{red}{3} &\textcolor{red}{3} &\textcolor{red}{3} &\textcolor{red}{3} &\textcolor{red}{3} &\textcolor{red}{3} &\textcolor{red}{3} &\textcolor{red}{3} &\textcolor{red}{3}&\textcolor{red}{3}&\textcolor{red}{3}&\textcolor{red}{3}&\textcolor{red}{3}\\
8  &2 &2 &2 &\textcolor{red}{3} &\textcolor{red}{3} &2 &\textcolor{red}{3} &\textcolor{red}{3} &\textcolor{red}{3} &\textcolor{red}{3} &\textcolor{red}{3} &\textcolor{red}{3} &\textcolor{red}{3} &\textcolor{red}{3} &\textcolor{red}{3} &\textcolor{red}{3}&\textcolor{red}{3}&\textcolor{red}{3}&\textcolor{red}{3}&\textcolor{red}{3}\\
9  &2 &2 &2 &\textcolor{red}{3} &2 &\textcolor{red}{3} &2 &2 &\textcolor{red}{3} &\textcolor{red}{3} &\textcolor{red}{3} &\textcolor{red}{3} &\textcolor{red}{3} &\textcolor{red}{3} &\textcolor{red}{3} &\textcolor{red}{3}&\textcolor{red}{3}&\textcolor{red}{3}&\textcolor{red}{3}&\textcolor{red}{3}\\
10 &2 &2 &2 &2 &\textcolor{red}{3} &\textcolor{red}{3} &\textcolor{red}{3} &2 &\textcolor{red}{3} &\textcolor{red}{3} &\textcolor{red}{3} &\textcolor{red}{3} &\textcolor{red}{3} &\textcolor{red}{3} &\textcolor{red}{3} &\textcolor{red}{3} &\textcolor{red}{3}&\textcolor{red}{3}&\textcolor{red}{3}&\textcolor{red}{3}\\
\bottomrule
\end{tabular}
\end{threeparttable}\end{adjustwidth}
\vspace{3ex}
Table~\ref{tbl:by-apm-3d-pmCnt} gives the position of the principal minor when a zero minor was found, and it can be seen that for field sizes $36, 37$ and $38$, the principal minor is towards the left side of the matrix $\mathcal{A}$.
As the field size increases, the position of the principal minor moves towards the end of the matrix. For a given deviation when the position of the principal minor is towards the right side of $\mathcal{A}$ then we increase the deviation count.
If the deviation count is not increased then the algorithm will take more than one kernel to solve the discrete logarithm.
In the case of two deviations the complexity of the algorithm is $O(l^5)$.
For three deviations the complexity is $O(l^7)$.
In the general case of $n$ deviations the complexity is $O(l^{2n+1})$.

\begin{adjustwidth}{-2.5 cm}{-2.5 cm}\centering\begin{threeparttable}[!htp]
\caption{APM : principal minor position}\label{tbl:by-apm-3d-pmCnt}
\footnotesize
\begin{tabular}{c|cccccccccccccccccccccc}\toprule
\multirow{2}{*}{Attempts}&  \multicolumn{20}{c}{Field size in bits} \\
 &27 &28 &29 &30 &31 &32 &33 &34 &35 &36 &37 &38 &39 &40 &41 &42 &43 &44&45&46&47\\\midrule
1  &13 &5 &10 &10 &1 &23 &1 &1 &1 &1 &4 & 39 &13 &10 &34 &13 &46 &63&5 &12&73\\
2  &6 &7 &21 &22 &90 &74 &1 &1 &1 &3 &1 &1 &8 &9  &38 &31 &23&82&1 &8&85\\
3  &4 &49 &1 &9 &20 &66 &60 &1 &1 & 32 &1 &2 &3 &13 &4 &25 &20&33&109 &74 &46\\
4  &2 &6 &12 &55 &1 &1 &1 &1 &1 &1 &2 &1 &4 &5 &19 &33 &20&21&43 &50&36\\
5 &11 &1 &13 &1 &1 &47 &1 &1 & 20 &2 &2 &3 &2 &5 &19 &26 &17&87&88 &14&76\\
6  &3 &4 &2 &2 &1 &27 &1 &1 &2 &1 &2 &2 &1 &13 &7 &47 &21&56&26 &26&16\\
7  &2 &4 &7 &1 &66 &1 &1 &1 &1 &1 &6 &1 &1 &9 &1 &56 &63&44&40 &71&61\\
8 &21 &14 &35 &26 &1 &1 &47 &1 &1 &2 &2 &3 &1 &31 &1 &22 &19&10&36 &27&9 \\
9  &1 &10 &6 &9 &1 &76 &1 & 25 & 95 &1 &2 &2 &1 &17 &7 &41 &4&26&63 &17&9\\
10  &3 &15 &4 &11 &24 &1 &1 &1 & 90 &3 &1 &3 &2 &8 &4 &6 &84&121&58 &63&29\\
\bottomrule
\end{tabular}
\end{threeparttable}\end{adjustwidth}
\vspace{3ex}

\subsection{Comaparing APM and GESC}
Table~\ref{tbl:all-apm} gives a comparison between the average kernel count for the almost principal minor algorithms, and Gaussian elimination Schur complement algorithm. 
Figure~\ref{fig:comparision} gives a logarithmic scaled graph for data given in Table~\ref{tbl:all-apm}. 
For the Gaussian elimination Schur complement experiment, the kernel size was amplified three times, while for the almost principal minor algorithm, the kernel size was not amplified. When the kernel size is not amplified the number of kernel required increases. So the graph in Figure~\ref{fig:comparision} is not apples to apples. It will only get worse if we do an apple to apple comparison.

When we compare these two experiments, one can argue that the almost principal minor algorithm is testing more minors than the Gaussian elimination Schur complement algorithm. That is the reason for the speed-up. However, one must keep this in mind that to to be able to find minors means that there exist zero minors in $\mathcal{A}$. Thus testing more minors is producing better results.

This suggests that the almost principal minor algorithm uses a more comprehensive set of initial minors than Gaussian elimination Schur complement. The data suggests that the growth of number of kernels in Gaussian elimination Schur complement is much slower than Gaussian elimination Schur complement algorithm. Though from these experiments it is not possible to say for sure if what the future growth will be, we think that this can be a viable line of investigation in solving the discrete logarithm problem.

From the data available one can probably postulate that the almost principal minor algorithm has a polynomial time complexity in size of the matrix $\mathcal{A}$. This can yield to a polynomial time algorithm to solve the discrete logarithm problem in the future.

\begin{adjustwidth}{-2.5 cm}{-2.5 cm}\centering\begin{threeparttable}[!htp]
\caption{APM Vs Gaussian elimination Schur complement : Average kernel count}\label{tbl:all-apm}
\footnotesize
\begin{tabular}{r|cccccccccccccccccc}\toprule
\multirow{2}{*}{Algorithm}&  \multicolumn{16}{c}{Field size in bits} \\
 &30 &31 &32 &33 &34 &35 &36 &37 &38 &39 &40 &41 &42 &43 & 44 & 45& 46 & 47\\\midrule
APM &1 &1 &1 &1 &1 &1 &1 &1 &1 &1 &1 &1 &1 &1 &1 & 1.2 & 1.6 & 1.8 \\
GESC &1 &1 &1 &1 &1 &1 &1 &1 &1 &1.2 &2.8 &3.1 & 6.6 & 11.3 & 17.1 & 53.8 &40.2 & 126 \\
\bottomrule
\end{tabular}
\end{threeparttable}\end{adjustwidth}
\begin{center}
\begin{figure}
 \includegraphics[width=17.9cm]{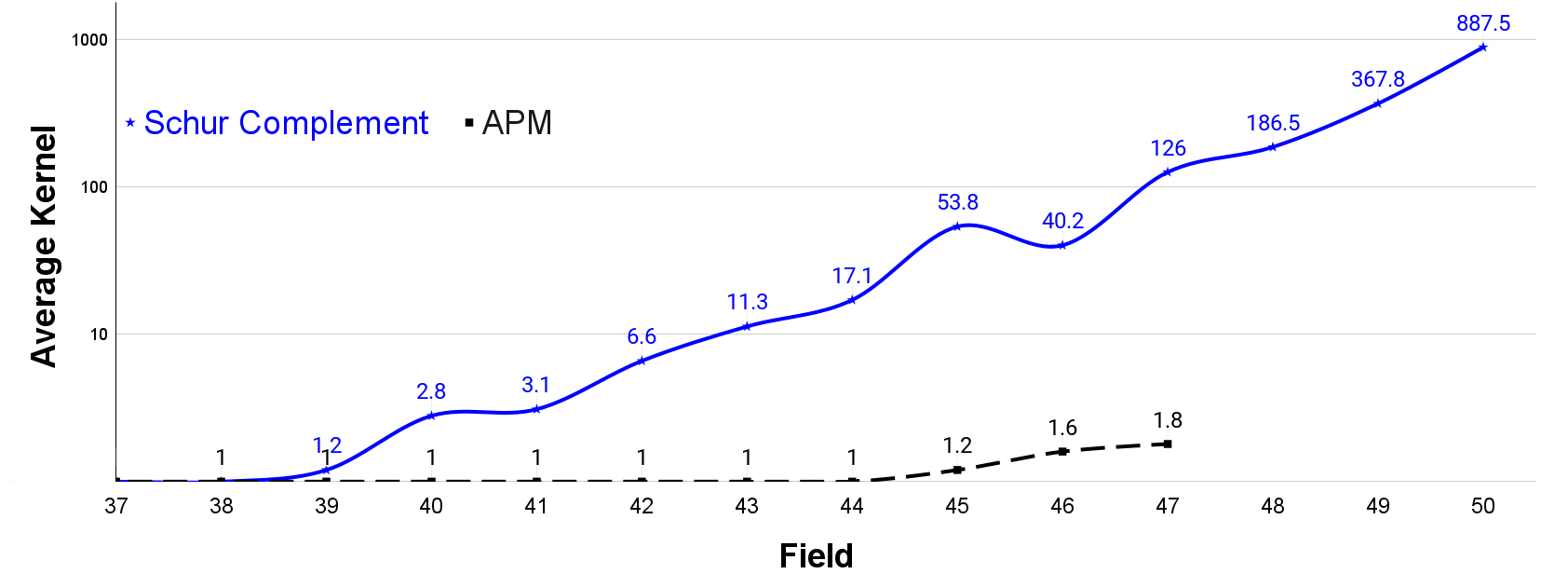}
\caption{GESC Vs APM }
\label{fig:comparision}
\end{figure}
\end{center}
\section{Conclusions}
In this paper, we continued developing a novel method to solve the elliptic curve discrete logarithm problem~\cite{first,second}. Our method is unique as it uses minors to solve the discrete logarithm problem. Furthermore, we tried to establish the initial minor conjecture on a firm footing. We find this conjecture interesting.

At the heart of our method is partition of $m$ modulo $p$ into $d$ parts. The process is simply to find integers $n_i$ for $i=1,2,\ldots,d$ where each $0<n_i<p$ such that $\sum_i{n_i}=0\mod p$. However, not any partition will do. Since we want to solve the discrete logarithm problem, both $\mathrm{P}$ and $\mathrm{Q}$ should be involved. Which means that $m$ should divide some of these $n_i$, but not all. Since, we do not know $m$, this is a hard problem. The way we handled this problem is by going to curves $\mathcal{C}$ intersecting the elliptic curves $\mathcal{E}$ at $d$ points. The extra points theorem gave us the leverage to move into zero minors.

Now once we see this process as a whole, then one partition of $m$ will give rise to several partitions of $m$. One such possibility is that for a non-zero $c$, such that, $0<c<p$. Now consider $n_i=n_i-c$ and $n_j=n_j+c$ for two different $i$ and $j$, which are not divisible by $m$. This will give rise to a different set of points on $\mathcal{E}$ and subsequently another partition of $m$. Of course, there will be other partitions related to the original partition that have no explicit and easy description like the one above.

But, our algorithm is probabilistic in nature. We have no control on the points $n_i$ that we choose. They are chosen at random. On the other hand, there are simply so many related partitions that some of them are bound to occur. In other words, once we see one zero minor, there will be many zero minors related to the original zero minor. This was also observed experimentally.

In the following experiment we took the finite field of size $2^{25}$. We computed a left-kernel $\mathcal{K}$ and then $\mathcal{A}$ for an appropriate choice of elliptic curve and points $P_i$ and $Q_j$ as described earlier. The goal was to find all zero minors in $\mathcal{A}$. It turned out to be impossible. So we used a short cut. We took contiguous principal minors one after the other of size $2$. For each one of those principal minors we took all possible $2,3$ and $4$ deviations. We checked for zero minor in those almost principal minors. The result was somewhat stunning and is presented in this table
\begin{center}
\begin{tabular}{ r | c c c c}
 & 2 deviation & 3 deviation & 4 deviation & total \\ 
 \hline
 average no.~zero minor & 1 & 230 & 70498 & 70729 \\ 
 standard deviation & & 15 & 349 &   351
\label{ZM-table}
\end{tabular}
\end{center}
where the average is over $66$ principal minors of size $2$.

One would expect that there would be a pattern in these zero minors. Though we have not found any pattern. But, a pattern would be interesting and that will shed light to the initial minor conjecture.

\bibliographystyle{abbrv}
\bibliography{ref}
\end{document}